\documentclass[a4paper,11pt,nopdfoutputerror]{quantumarticle}
\usepackage{amsmath,amssymb,amsthm}
\usepackage{booktabs}
\usepackage{array}
\usepackage{graphicx}
\usepackage[colorlinks=true,linkcolor=blue!60!black,citecolor=blue!60!black,urlcolor=blue!60!black]{hyperref}

\newtheorem{theorem}{Theorem}
\newtheorem{proposition}[theorem]{Proposition}
\newtheorem{lemma}[theorem]{Lemma}
\newtheorem{fact}[theorem]{Fact}
\newtheorem{corollary}[theorem]{Corollary}
\theoremstyle{remark}
\newtheorem*{remark}{Remark}

\newtheorem{alemma}{Lemma}

\newtheorem{elemma}{Lemma}

\newtheorem*{etheorem}{Theorem E}
\newtheorem*{acorollary}{Corollary (the transform bound)}

\newcommand{\Dp}{\mathcal{D}_p}
\newcommand{\Tr}{\mathrm{Tr}}
\newcommand{\id}{\mathrm{id}}
\newcommand{\ent}{\mathsf{H}}

\title{A certified lower bound on the quantum-capacity threshold of the depolarizing channel}
\author{Artus Krohn-Grimberghe}
\affiliation{Percivio Ltd.}
\date{August 14, 2026}

\begin{document}

\maketitle

\begin{abstract}
The noise threshold below which the qubit depolarizing channel retains
positive quantum capacity has been studied since 1996. The classic
constructions come with exact finite formulae, but in every reported
threshold to date --- most recently the record of Agarwal et
al.~\cite{agarwal2026} --- the final step, the sign of the coherent
information at the reported point, is a floating-point evaluation. We
give the first exact-arithmetic, independently machine-checkable
positivity proofs in this regime: an explicit 45-copy rank-two input
state $A$, published as a 1{,}472-byte witness, with certified $I_c > 0$
at the exact rational $p = 16239/250000 = 0.064956$ (per-Pauli
convention; total error rate $3p = 0.194868$) --- beyond both the best
value printed in Ref.~\cite{agarwal2026} ($0.064657$) and the strongest
state in the same authors' public repository ($\approx 0.064911$), both
numerical. Because the depolarizing family is a
semigroup under composition, fixed-input coherent information is
nonincreasing in $p$ on $[0,1/4]$, so a single certified point extends to
the entire interval below it. The same machinery confines the positivity
boundaries of $A$ and of the strongest public state of
Ref.~\cite{agarwal2026} to disjoint rational intervals separated by more
than $1/25000$ --- to our knowledge the first proven ordering of the
positivity boundaries of two explicit competing code states. Every computational claim in the paper reduces to a
finite list of big-integer comparisons, checkable by a dependency-free
few-hundred-line verifier whose soundness rests on three self-contained
half-page lemmas. Payloads, certificates, and verifier accompany the paper
as a supplementary artifact.
\end{abstract}

\section{Introduction}
\label{sec:intro}

How much depolarizing noise can a qubit channel suffer before its quantum
capacity vanishes? Thirty years of work has pushed the known-positive
region outward:

\begin{table*}[t]
\centering
\resizebox{\textwidth}{!}{%
\begin{tabular}{llll}
\toprule
Year & Bound (per-Pauli $p$) & Source & Evidence type \\
\midrule
1996 & 0.063096 & hashing (random codes)~\cite{bennett1996} & analytic formula, numerical root \\
1996 & 0.063452 & Shor--Smolin concatenation~\cite{shorsmolin1996} (re-evaluated in~\cite{fernwhaley2008}) & closed form, numerical crossing \\
1998 & 0.06352 & DiVincenzo--Shor--Smolin~\cite{dss1998} & closed form, numerical crossing \\
2007 & 0.063626 & Smith--Smolin~\cite{smithsmolin2007} & closed form, numerical crossing \\
2008 & 0.0637675 & Fern--Whaley~\cite{fernwhaley2008} & Monte Carlo ($\pm 6\times 10^{-8}$) \\
2026 & 0.064657 & Agarwal et al., symmetry-enhanced search~\cite{agarwal2026} & floating-point eigensolver \\
this work & \textbf{0.064956} & state $A$, certified & \textbf{exact-rational certificate} \\
\bottomrule
\end{tabular}%
}
\caption{Reported positivity points for the quantum capacity of the qubit
depolarizing channel, per-Pauli convention (multiply by three for the
total-error-rate convention). The last column describes the status of the
printed boundary value; the constructions and formulae behind the classic
entries are themselves rigorous. The 1996 Shor--Smolin paper prints
fidelity $f = 0.8096$, i.e.\ $p \approx 0.06347$; the tabulated value is
Fern and Whaley's later exact enumeration of the same code.}
\label{tab:history}
\end{table*}

We use the per-Pauli convention throughout: $\Dp(\rho) = (1-3p)\rho +
p(X\rho X + Y\rho Y + Z\rho Z)$, so the total probability of a nonidentity
Pauli is $3p$; results quoted from the literature are converted once,
here, and never again. Much of the literature writes the same family as
$D_q(\rho) = (1-q)\rho + \tfrac{q}{3}(X\rho X + Y\rho Y + Z\rho Z)$ with
$q = 3p$ the total error probability --- a convention with its own
virtue, since $1-q$ is then the singlet fraction: the fidelity of the
state obtained by sending half of a maximally entangled pair. In that
convention our certified point $p = 16239/250000$ reads
$q = 48717/250000 = 0.194868$.

The last column is the point of this paper, and it is a statement about
printed boundary values, not about the classic constructions, which are
rigorous mathematics: DiVincenzo, Shor, and Smolin derived closed-form
expressions for their concatenated codes~\cite{dss1998}, and positivity
above the hashing zero has been believed since 1996 on the strength of
such formulae. What no prior entry supplies is a verifiable sign: in
each case the final step is a numerical evaluation --- the decimal root
of an analytic formula, a floating-point crossing of a closed form, or a
Monte Carlo extrapolation. The incumbent itself reports its
coherent-information computation as ``accurate up to at least 10 decimal
places''~\cite{agarwal2026} --- an accuracy claim, not a certificate. Coherent-information margins in this regime
are of order $10^{-7}$ bits --- far below anything a double-precision
eigensolver can settle on its own, and the superadditivity phenomena
that make the threshold interesting live precisely in such margins. Ours
is the first entry whose sign carries a machine-checkable proof, and the
paper is organized so that the reader can verify it without trusting the
authors, their code, or their arithmetic.

The eighteen years between 2008 and 2026 were not quiet:
neural-network ans\"atze~\cite{bauschleditzky2020}, systematic sweeps of
the Pauli simplex~\cite{bauschleditzky2021}, genetic and particle-swarm
searches~\cite{sidhardh2022}, low-noise expansions~\cite{lls2018},
perturbative criteria~\cite{wu2025}, and further stabilizer studies by
the incumbent group~\cite{agarwal2026mar} improved other Pauli families
without moving the depolarizing record, which fell only to the
symmetry-enhanced search of Ref.~\cite{agarwal2026}.

\paragraph{How to verify this paper.} Three steps, in increasing depth:
\begin{enumerate}
\item \emph{Five minutes, using the incumbents' code, not ours.} Run the
two payloads of the artifact through the public floating-point evaluator
of Ref.~\cite{agarwal2026} and observe the four endpoint signs of
Fact~\ref{fact:signs}.
\item \emph{One command.} Run the supplied pure-integer verifier on the
four certificates: standard library only, \texttt{int} and
\texttt{Fraction}, no floats, no eigensolver, \texttt{python3 -O}-safe,
bit-reproducible.
\item \emph{One afternoon.} Check
Propositions~\ref{prop:affine}--\ref{prop:composition} and
Theorem~\ref{thm:confinement} by hand (elementary algebra and the
intermediate value theorem), and, if desired, the three soundness lemmas
of Appendix~\ref{app:lemmas} --- half a page each.
\end{enumerate}
Nothing else in the paper is load-bearing.

\paragraph{Contributions.} (i) A certification framework that reduces
coherent-information sign claims at 45 copies to big-integer comparisons
(Tier 3; Section~\ref{sec:certified}, Appendices~\ref{app:lemmas}--\ref{app:verifier}).
(ii) Monotonicity of fixed-input coherent information on $[0,1/4]$ by
semigroup composition plus data processing (Tier 2;
Section~\ref{sec:monotonicity}). (iii) Confinement of the positivity
boundaries of two explicit states to disjoint exact rational intervals,
with a proven separation (Tier 1 given (i)--(ii);
Section~\ref{sec:results}). (iv) The record certified positivity point
$p = 16239/250000$ (Corollary~\ref{cor:threshold}).

Section~\ref{sec:notclaimed} states explicitly what we do not claim.

\section{Setting and conventions (Tier 1)}
\label{sec:setting}

$\Dp(\rho) = (1-3p)\rho + p(X\rho X + Y\rho Y + Z\rho Z)$ is CPTP for
$p \in [0,1/3]$. Entropies are in bits. For a channel $\mathcal{N}$ and
input $\rho$ with purification $\psi_\rho$, the coherent information is
$I_c(\rho,\mathcal{N}) = S(\mathcal{N}(\rho)) -
S((\id\otimes\mathcal{N})(\psi_\rho))$, and $Q(\mathcal{N}) \ge
I_c(\rho,\mathcal{N}^{\otimes n})/n$ for every $n$ and $\rho$ (LSD
theorem~\cite{lloyd1997,shor2002,devetak2005}).

Two 45-qubit states appear:
\begin{itemize}
\item $R$ --- the strongest public $n=45$ state of the
Ref.~\cite{agarwal2026} repository (payload SHA-256
\texttt{\small 9a94606d28cb787ebffdc0c4641562}\allowbreak\texttt{\small 3ae9e83d27df5de655f4114e86ca4cc1c6}),
credited to its authors. Our payload is byte-identical to the tensor
storage of that repository's published \texttt{n\_45\_ens.pt} file (as
retrieved 2026-08-14; the printed hash pins the exact bytes certified), so
the certified $R$ enclosures apply to their state exactly as published,
with no re-encoding.
\item $A$ --- this work (payload SHA-256
\texttt{\small 1df2e00c06b2614cd04a5bfc7bede0}\allowbreak\texttt{\small 40da2d3fc1dd49e7138e201c101aab266f}).
\end{itemize}

Both payloads are 1{,}472 bytes: little-endian complex128 coefficients in
a $2\times 46$ layout, the rows of a rank-two mixture of
permutation-symmetric (Dicke-basis) pure states, interpreted exactly as
dyadic rationals. That interpretation --- the bytes, not any
floating-point neighborhood of them --- is the object all certificates
refer to. The state is a witness; how it was found is irrelevant to every
claim that follows (see the contribution statement).

We write $f_X(p) = I_c(\rho_X, \Dp^{\otimes 45})$ in bits for
$X \in \{R, A\}$.

\section{Elementary structure of the depolarizing family (Tier 1)}
\label{sec:structure}

\begin{proposition}[affine form]
\label{prop:affine}
For every $2\times 2$ matrix $T$:\; $\Dp(T) = (1-4p)\,T + 2p\,\Tr(T)\,I$.
\end{proposition}

\begin{proof}
Both sides are linear in $T$; check the four basis cases
$T \in \{I,X,Y,Z\}$: Pauli conjugation fixes $I$ and gives
$\sigma_i \mapsto (1-4p)\sigma_i$ since each nonidentity Pauli
anticommutes with exactly two of $X,Y,Z$.
\end{proof}

\begin{proposition}[composition]
\label{prop:composition}
$\mathcal{D}_q \circ \Dp = \mathcal{D}_{p+q-4pq}$. Consequently, for
$0 \le p_1 \le p_2 \le 1/4$, $\mathcal{D}_{p_2} = \mathcal{D}_q \circ
\mathcal{D}_{p_1}$ with $q = (p_2-p_1)/(1-4p_1) \in [0,1/4]$.
\end{proposition}

\begin{proof}
By Proposition~\ref{prop:affine} the family acts on the traceless part by
the scalar $1-4p$, and $(1-4p)(1-4q) = 1-4(p+q-4pq)$. Solving
$p_2 = p_1 + q - 4p_1 q$ for $q$ gives the stated value, which lies in
$[0,1/4]$ because $0 \le p_2 - p_1 \le 1/4 - p_1 \le 1 - 4p_1$ (times
$1/4$).
\end{proof}

\begin{remark}[scope]
At $p = 1/4$ the family passes through the completely depolarizing
channel; beyond it $1-4p < 0$ and the ordering argument genuinely fails.
Witness: for the maximally mixed single-qubit input,
$$I_c = 1 - H(1-3p,\,p,\,p,\,p)$$
rises from $-1$ at $p=1/4$ to $1-\log_2 3$ at $p=1/3$. All monotonicity claims in this paper are restricted to
$[0,1/4]$.
\end{remark}

\section{Monotonicity (Tier 2 --- the imported theorem)}
\label{sec:monotonicity}

\begin{lemma}
\label{lem:monotone}
For every fixed 45-qubit input state $\rho$, the map $p \mapsto
I_c(\rho, \Dp^{\otimes 45})$ is nonincreasing on $[0,1/4]$.
\end{lemma}

\begin{proof}
For $p_1 \le p_2$ in $[0,1/4]$, Proposition~\ref{prop:composition} gives
$\mathcal{D}_{p_2}^{\otimes 45} = \mathcal{D}_q^{\otimes 45} \circ
\mathcal{D}_{p_1}^{\otimes 45}$ with $\mathcal{D}_q^{\otimes 45}$ CPTP.
Coherent information does not increase under CPTP post-processing of the
channel output~\cite{schumachernielsen1996}.
\end{proof}

Both ingredients are classical; we claim the use, not the ingredients.
Data processing is the only theorem this paper imports
(Appendix~\ref{app:derivation} removes what would otherwise be a second
import).

\section{Four certified signs (Tier 3 --- the computation)}
\label{sec:certified}

\paragraph{Block decomposition.} For a rank-two permutation-symmetric
input, the $2^{45}$-dimensional computation collapses. Stated for general
$n$ (proved self-contained in Appendix~\ref{app:derivation}): for
$|\Psi\rangle = N^{-1/2}\sum_r |r\rangle_{\mathrm{ref}} \otimes \sum_k
\psi[r][k]\,|D_k^n\rangle$, the joint output $(\id_{\mathrm{ref}} \otimes
\Dp^{\otimes n})(|\Psi\rangle\langle\Psi|)$ is unitarily equivalent to
$\bigoplus_f B_f \otimes I_{m_f}$, where $f$ ranges over
$\lceil n/2\rceil,\dots,n$, each $B_f$ is an explicit
$2(2f-n+1)$-dimensional Hermitian matrix with entries rational in $p$ (up
to one square-root normalization), and $m_f =
\binom{n+1}{f+1}(2f-n+1)/(n+1)$. Hence $I_c = \sum_f
m_f\,[H(\Tr_{\mathrm{ref}} B_f) - H(B_f)]$, and at $n=45$ the largest
block is 92-dimensional. This decomposition was developed by Bhalerao
and Leditzky~\cite{bhalerao2025} for computing the coherent information
of permutation-invariant codes; Ref.~\cite{agarwal2026} extended it to
arbitrary states of the symmetric subspace, and its public evaluator
implements it --- that is exactly what makes verification step (i)
meaningful. Permutation-invariant codes themselves are much older, going
back to Pollatsek and Ruskai~\cite{pollatsekruskai2004} and
Ouyang~\cite{ouyang2014}; a general error-correction theory for them and
a Schur--Weyl framework for permutation-invariant optimization appeared
in parallel with this line of
work~\cite{ouyangbrennen2026,berghparentin2026}. Our
proof in Appendix~\ref{app:derivation} is independent and self-contained,
so the only theorem the paper imports remains data processing.

\begin{fact}[certified]
\label{fact:signs}
With the exact dyadic payload semantics of Section~\ref{sec:setting}:
\begin{align*}
f_R(4057/62500) > 0, &\quad f_R(16229/250000) < 0,\\
f_A(16239/250000) > 0, &\quad f_A(203/3125) < 0.
\end{align*}
\end{fact}

Each sign is established by a certificate in the supplementary artifact;
the verifier reduces each to a finite list of big-integer comparisons
(format: Appendix~\ref{app:format}; soundness: Appendix~\ref{app:lemmas};
verifier: Appendix~\ref{app:verifier}). The certified enclosures, as
exact rationals rounded outward to six figures here, are:

\begin{table*}[t]
\centering
\begin{tabular}{llll}
\toprule
State & $p$ (exact) & $p$ (decimal) & $I_c$ enclosure (bits) \\
\midrule
$R$ & $4057/62500$ & $0.064912$ & $[5.79717\times 10^{-8},\; 6.17746\times 10^{-8}]$ \\
$R$ & $16229/250000$ & $0.064916$ & $[-1.36795\times 10^{-7},\; -1.32951\times 10^{-7}]$ \\
$A$ & $16239/250000$ & $0.064956$ & $[8.38952\times 10^{-8},\; 8.70755\times 10^{-8}]$ \\
$A$ & $203/3125$ & $0.064960$ & $[-4.35188\times 10^{-8},\; -4.03627\times 10^{-8}]$ \\
\bottomrule
\end{tabular}
\caption{The four certified enclosures of Fact~\ref{fact:signs}, exact
rationals rounded outward to six figures.}
\label{tab:enclosures}
\end{table*}

\paragraph{What is recomputed vs.\ hinted.} The verifier recomputes, in
exact rational arithmetic from the payload bytes alone, the block
matrices $B_f$ at the four rational values of $p$. The only untrusted
hints a certificate supplies are a rounded binary64 center matrix and a
candidate eigenbasis per block; each hint is fenced (conversion check,
Gram check, polar-factor bound, transform bound, disc isolation) before
anything is concluded from it, and every analytic inequality is
pre-cleared of denominators so that the final comparisons involve
integers only. Nothing load-bearing escapes a check.

\section{Main results (Tier 1, standing on Sections 3--5)}
\label{sec:results}

\begin{theorem}[confinement and separation]
\label{thm:confinement}
On $[0,1/4]$:
\begin{enumerate}
\item every zero of $f_R$ lies in $(4057/62500,\; 16229/250000)$;
\item every zero of $f_A$ lies in $(16239/250000,\; 203/3125)$;
\item both zero sets are nonempty;
\item every zero of $f_A$ exceeds every zero of $f_R$ by more than
$1/25000$.
\end{enumerate}
\end{theorem}

\begin{proof}
By Lemma~\ref{lem:monotone}, $f_R$ is nonincreasing; by
Fact~\ref{fact:signs} it is positive at $4057/62500$ and negative at
$16229/250000$. Monotonicity forces $f_R > 0$ on $[0, 4057/62500]$ and
$f_R < 0$ on $[16229/250000, 1/4]$, so all zeros lie between the
endpoints; continuity of $f_R$ in $p$ and the intermediate value theorem
give at least one zero. The same argument delivers (2) and (3) for
$f_A$. For (4): $16239/250000 - 16229/250000 = 10/250000 = 1/25000$, and
every zero of $f_A$ is strictly above $16239/250000$ while every zero of
$f_R$ is strictly below $16229/250000$.
\end{proof}

\begin{corollary}[threshold]
\label{cor:threshold}
$Q(\Dp) \ge f_A(p)/45 > 0$ for every $p \le 16239/250000 = 0.064956$.
\end{corollary}

\begin{proof}
Lemma~\ref{lem:monotone} plus Fact~\ref{fact:signs} give $f_A(p) \ge
f_A(16239/250000) > 0$ on the stated interval; the LSD bound divides by
$n = 45$.
\end{proof}

Worded precisely: the best previously \emph{reported} positivity point in
this convention is $0.064657$ (numerical, Ref.~\cite{agarwal2026}); the
strongest value in the same authors' public repository is
$\approx 0.0649112$, also numerical. Corollary~\ref{cor:threshold} is, to
our knowledge, both the largest reported point and the first in this
regime established in exact arithmetic rather than by numerical
evaluation. The claim is about proof status and the exact point; we
make no claim of optimality over unpublished states
(Section~\ref{sec:notclaimed}).

\begin{remark}[full CPTP range]
The channel is antidegradable for $p \ge 1/12$ --- total error rate
$3p \ge 1/4$, the form in which the bound is usually
quoted~\cite{cerf1998,crs2008}, e.g.\ in the concurrent strong-converse
work of Kondra et al.~\cite{kondra2026} --- whence $Q = 0$ there by the
no-cloning bound~\cite{antidegradable}; it is entanglement-breaking on
$[1/6,1/3]$~\cite{hsr2003,ruskai2003} (strictly weaker for this purpose,
kept as a remark). Combined with Theorem~\ref{thm:confinement}, the
fixed-state positivity boundaries $\tau_R, \tau_A$ obey the same brackets
within the full CPTP range $[0,1/3]$.
\end{remark}

\section{What is not claimed}
\label{sec:notclaimed}

\begin{itemize}
\item No uniqueness of zeros: Theorem~\ref{thm:confinement} confines zero
\emph{sets}; we do not claim $f_R$ or $f_A$ crosses once.
\item No exact threshold and no optimality: $A$ is one explicit state; we
do not claim it is optimal within any family, nor that $16239/250000$
approximates the true capacity threshold from below by any stated gap.
\item No family-wide or all-state statement, and no capacity formula.
\item No zero-capacity claim anywhere below $1/12$.
\item No pointwise comparison $f_A > f_R$ across $[0,1/3]$.
\item Monotonicity is claimed only on $[0,1/4]$
(Section~\ref{sec:structure}, Remark).
\end{itemize}

\section{Discussion}
\label{sec:discussion}

The certified point stands at $0.064956$ against the no-cloning ceiling
$1/12 = 0.0833\ldots$; the gap remains large, and nothing here suggests
where in it the true threshold lies. On the converse side, we are not
aware of any printed zero-capacity threshold below the antidegradability
point $p = 1/12$: the flagged-extension upper bounds on
$Q(\mathcal{D}_p)$~\cite{fanizza2020,kianvash2022} tighten the capacity
inside the interval, and partial-order and low-noise methods constrain
it near the ends~\cite{hircheleditzky2023,lls2018}, without moving that
threshold. The contribution we expect to travel
is the method: coherent-information margins of order $10^{-7}$ bits are
exactly where superadditivity debates live and exactly where floating
point alone stops being evidence for a sign. The pattern used here --- exact rational
block arithmetic from a byte-level witness, untrusted hints fenced by
cheap checks, all inequalities flattened to integer comparisons,
soundness carried by three half-page lemmas --- applies verbatim to any
channel family with enough symmetry to collapse the block structure, and
in weakened form well beyond. Certified optimization loops that produce
such witnesses at scale are a natural next step.

Proof by certificate has precedents, and we claim the application, not
the pattern. Outside quantum information, computer-assisted proofs from
the four-color theorem to the Kepler conjecture were eventually reduced
to independently checkable objects~\cite{gonthier2008,hales2017}, SAT
solvers emit proofs that a small trusted checker
replays~\cite{wetzler2014} --- at the extreme, the Boolean Pythagorean
triples proof, whose $\sim$200-terabyte certificate is replayed by a
small independent checker~\cite{heule2016} --- and verified numerics
enclose spectra in interval arithmetic~\cite{rump1999}. Inside quantum
information, semidefinite-programming duality yields rigorous converse
bounds~\cite{wangfangduan2019}, rational certificates have recently
upgraded numerical bounds on quantum codes~\cite{anglesmunne2026} and,
more generally, exact rational bounds are now extracted from
floating-point semidefinite relaxations of quantum optimization
problems~\cite{naceur2025}, and proof assistants are beginning to
formalize the theory itself~\cite{leanqit2026,leanquantum2026}. What is new here is an
exact-arithmetic sign proof for depolarizing coherent information at
record noise, packaged so that a dependency-free verifier replays it
from the published bytes.

\section*{Acknowledgments}

State $A$ was obtained by numerical optimization initialized from the
public $n=45$ states of Agarwal, Kalra, Lee, Leung, Schaeffer, Sinha, and
Smith~\cite{agarwal2026}; we thank the authors for making their states
and evaluator public.

\section*{Contribution statement}

The sole author takes full responsibility for the entire content. Scope
of AI use in producing this work: the input state $A$ was found by an
AI-driven numerical search initialized from the public states credited
above; AI systems were further used for code production, for
calculations, for the mechanical derivation of proof steps, and for
drafting text under the author's direction. The verification chain
(Sections~\ref{sec:certified}--\ref{sec:results},
Appendices~\ref{app:lemmas}--\ref{app:derivation}) is independent of how
the state and the derivations were produced: every computational claim
reduces to big-integer comparisons checkable by the supplied
dependency-free verifier, and every imported mathematical statement is
either proved in the appendices or cited to the classical literature.

\appendix

\section{Three soundness lemmas}
\label{app:lemmas}

\begin{alemma}[inflated Gershgorin with component counting]
\label{lem:gershgorin}
Let $T$ be a $D\times D$ Hermitian matrix, and let $H = T + E$ be
Hermitian with $\|E\|_2 \le \varepsilon$. For each $i$ let $s_i \ge
\sum_{j\ne i} |T_{ij}|$ and set $\widehat G_i = [\,T_{ii} - s_i -
D\varepsilon,\; T_{ii} + s_i + D\varepsilon\,]$. Then every eigenvalue of
$H$ lies in $\bigcup_i \widehat G_i$, and if that union splits into
disjoint closed components, a component formed from $k$ intervals
contains exactly $k$ eigenvalues of $H$ (with multiplicity).
\end{alemma}

\begin{proof}
Every entry of $E$ obeys $|E_{ij}| = |e_i^* E e_j| \le \|E\|_2 \le
\varepsilon$. Gershgorin's theorem applied to $H$ (Hermitian, so discs
are real intervals) places each eigenvalue in some interval centered at
$H_{ii} = T_{ii} + E_{ii}$ with radius $\sum_{j\ne i} |H_{ij}| \le s_i +
(D-1)\varepsilon$; since $|E_{ii}| \le \varepsilon$, that interval lies
inside $\widehat G_i$. For the count, put $H(t) = \mathrm{diag}(H) +
t\,(H - \mathrm{diag}(H))$ for $t\in[0,1]$. Each $H(t)$ is Hermitian with
the same centers and row sums scaled by $t$, so its Gershgorin intervals
sit inside those of $H(1) = H$, hence inside the $\widehat G_i$. At
$t = 0$ the eigenvalues are the $H_{ii}$, one in each $\widehat G_i$.
Eigenvalues of a Hermitian family depend continuously on $t$; a
continuous path starting in one closed component of the union cannot
reach a disjoint component without leaving the union, which no eigenvalue
of any $H(t)$ does. Hence the count in each component is constant in $t$
and equals the number of intervals forming it.
\end{proof}

\begin{alemma}[polar-factor perturbation]
\label{lem:polar}
If a square matrix $V$ satisfies $\|V^*V - I\|_F \le \delta < 1$, then
$V$ is invertible, its polar factor $U$ (the unitary in
$V = U(V^*V)^{1/2}$) exists, and
$$\|V - U\|_2 \;\le\; \max\!\big(1 - \sqrt{1-\delta},\;
\sqrt{1+\delta} - 1\big).$$
\end{alemma}

\begin{proof}
Every eigenvalue of the Hermitian matrix $V^*V$ differs from $1$ by at
most $\|V^*V - I\|_2 \le \|V^*V - I\|_F \le \delta$, so every singular
value $\sigma$ of $V$ lies in $[\sqrt{1-\delta}, \sqrt{1+\delta}]$ with
$\sqrt{1-\delta} > 0$; in particular $V$ is invertible and $U$ exists.
Then $V - U = U\big((V^*V)^{1/2} - I\big)$ and, since $U$ is unitary,
$\|V - U\|_2 = \|(V^*V)^{1/2} - I\|_2 = \max_\sigma |\sigma - 1|$, which
is at most the stated maximum.
\end{proof}

\begin{acorollary}
Let $C$ be the exact Hermitian block, $B$ the recorded binary64 center
with $\|C - B\|_F \le M$, $V$ the recorded candidate basis with Gram
defect $\delta$, $U$ its polar factor, and $\rho$ the bound of
Lemma~\ref{lem:polar}. Since $U^* C U$ has the spectrum of $C$, and
\begin{multline*}
\|U^* C U - V^* B V\|_2 \le \|C - B\|_2\\
+ \|U^* B U - V^* B V\|_2
\le M + \rho\,\|B\|_F\,(1 + \|V\|_2),
\end{multline*}
(expanding $U^*BU - V^*BV = U^*B(U-V) + (U-V)^*BV$ and using $\|B\|_2 \le
\|B\|_F$), the spectrum of $C$ is that of the Hermitian matrix
$T = V^* B V$ perturbed by at most $\varepsilon = M +
\rho\|B\|_F(1+\|V\|_2)$ --- exactly the inflation
Lemma~\ref{lem:gershgorin} consumes. This is the certificate's
\texttt{transform\_error} field.
\end{acorollary}

\begin{alemma}[entropy enclosure]
\label{lem:entropy}
Let a positive-semidefinite Hermitian block with trace at most $1$ have
its spectrum confined, by Lemmas~\ref{lem:gershgorin}--\ref{lem:polar},
to disjoint components $[a_c,b_c]\cap[0,1]$ containing $k_c$ eigenvalues
each. Write $h(x) = -x\log_2 x$ on $[0,1]$ (with $h(0) = h(1) = 0$). Then
the von Neumann entropy $S = \sum_i h(\lambda_i)$ satisfies
$$\sum_c k_c\,\underline h_c \;\le\; S \;\le\; \sum_c k_c\,\overline h_c,$$
where, for a component with endpoints $a \le b$: $\underline h =
\min(h(a), h(b))$; $\overline h = h(b)$ if $b \le \underline{x}$,
$\overline h = h(a)$ if $a \ge \overline{x}$, and $\overline h = 1$
otherwise, with $\underline{x} = 367879/10^6 < 1/e < 367880/10^6 =
\overline{x}$ (an inequality the verifier proves at startup from the
directed logarithm bounds below).
\end{alemma}

\begin{proof}
Clipping to $[0,1]$ preserves containment because the true eigenvalues of
a PSD block of trace $\le 1$ lie in $[0,1]$. $h$ is concave on $[0,1]$
with $h'(x) = -(\ln x + 1)/\ln 2$, so $h$ increases on $(0,1/e)$,
decreases on $(1/e,1)$, and attains its maximum $h(1/e) = \log_2(e)/e <
0.5308 < 1$. On an interval avoiding $1/e$, $h$ is monotone, so its range
is between the endpoint values; on an interval that may contain $1/e$,
the minimum is still at an endpoint (concavity) and the maximum is at
most $h(1/e) < 1$, so the cap $\overline h = 1$ is valid; the four
shipped certificates remain sign-decisive with it. Summing over
components with their counts $k_c$ gives the display. Endpoint values
$h(a), h(b)$ at rational arguments are enclosed by directed rational
bounds on $\log a$: write $a = v\cdot 2^{-k}$ with $v\in[1,2]$, and use
\begin{gather*}
\log v = 2\sum_{k=0}^{95} \frac{t^{2k+1}}{2k+1} + R, \qquad
t = \frac{v-1}{v+1},\\
0 \le R \le \frac{2\,t^{193}}{193\,(1-t^2)},
\end{gather*}
where every summand is a positive rational and the tail bound is a
geometric-series comparison; floor/ceil each partial quantity on the
$10^{-70}$ grid in the direction that widens the enclosure, and divide by
an enclosure of $\ln 2$ obtained from the same series. All operations
are rational, all roundings directed, so the resulting bounds on $h$ are
rigorous.
\end{proof}

\section{Certificate format}
\label{app:format}

Each of the four endpoints ships with one certificate: a single JSON file
that records the complete enclosure chain of the evaluator run at that
endpoint, in exact arithmetic, so that the verifier of
Appendix~\ref{app:verifier} can replay every step without floating point
and without trusting the evaluator. This appendix specifies the format;
the normative statement is the schema itself
(\texttt{physics-proofs/}\allowbreak\texttt{flattened-rank-two-endpoint/}\allowbreak\texttt{v1}, pinned by the
certificate's first field).

\paragraph{Exactness conventions.} Five primitive encodings appear, all
built from canonical decimal integer strings (no signs on denominators,
no leading zeros): an integer \texttt{I}; a reduced rational \texttt{Q}
as a numerator/positive-denominator pair; a dyadic \texttt{D} as a
mantissa/exponent pair meaning $m\cdot 2^e$, with the zero dyadic written
\texttt{["0","0"]} and every nonzero mantissa odd (so each dyadic has
exactly one encoding); and Gaussian variants \texttt{GQ} and \texttt{GD}
giving real and imaginary parts. Matrices are stored row-major with
explicit dimensions. Because every value is a decimal string, the
certificate is independent of any binary floating-point format, and
canonical-form violations (an unreduced fraction, an even mantissa, a
re-encoded zero) are parse errors, not tolerated variants.

\paragraph{Contents.} The top level pins the schema name, $n = 45$, the
raw 1472-byte payload as hex together with its SHA-256, the exact channel
parameter $p$ as a reduced fraction, and the four working-precision
constants (square-root witness precision 224 bits; 96 logarithm series
terms; log grid at $10^{-70}$; entropy accumulation at $10^{-50}$). Then
come the twenty-three sector blocks $f = 23,\dots,45$ in order, each
carrying: the sector's degree $d = 2f - 45$ and Specht multiplicity
$m_f$; the exact block-trace enclosure; and two spectral records --- one
for the joint block $B_f$ (dimension $2(d+1)$) and one for its reference
partial trace (dimension $d+1$). A spectral record stores the rounded
dyadic center matrix $B$, the candidate eigenbasis $V$ (the only two
load-bearing hints in the entire certificate;
Appendix~\ref{app:verifier} explains how they are fenced), the ten
rational bounds of the Appendix~\ref{app:lemmas} chain (Frobenius
conversion and total errors, Gram defect, the two square-root witnesses,
polar distance, norm bounds, and the assembled transform error), the
per-row off-diagonal sums with their inflated Gershgorin discs, the
merged spectral components with multiplicities, clipped intervals,
endpoint entropy rationals and the $1/e$-cap flag, and the resulting
scaled entropy integers. The file closes with the global trace enclosure,
the six final scaled integers (directed $10^{-50}$-scaled lower and upper
bounds for the output entropy, the joint entropy, and the coherent
information), and the classification \texttt{POS} or \texttt{NEG}.

\paragraph{Sizes.} For $n = 45$ one matrix family (all blocks, joint plus
partial trace) has $\sum_d \big(d^2 + (2d)^2\big) = 86{,}480$ complex
entries; recording both $B$ and $V$ gives $172{,}960$ Gaussian dyadics,
which dominates the file. Result rationals reach a few hundred decimal
digits; all certificate operands stay below roughly 6{,}000 bits, and the
verifier enforces a hard 20{,}000-bit cap on every parsed integer
together with array-size caps derived from the fixed $n = 45$ block
topology, so a malicious certificate cannot make the verifier allocate
unbounded memory. Measured per-endpoint file sizes: 17{,}444{,}604
bytes ($R$ at $p = 4057/62500$), 17{,}442{,}411 ($R$ at $16229/250000$),
17{,}414{,}593 ($A$ at $16239/250000$), and 17{,}402{,}776 ($A$ at
$203/3125$). Certificate SHA-256s are listed in
Appendix~\ref{app:repro} and in the artifact manifest.

\section{The verifier}
\label{app:verifier}

The verifier is a single Python file in the artifact
(\texttt{verifier/verify.py}) with no dependencies outside the standard
library. Its arithmetic is \texttt{int} and \texttt{fractions.Fraction}
only; no float is ever constructed in a load-bearing step, and no
\texttt{assert} statement guards a check (so running under
\texttt{python3 -O} cannot disable anything). It reads one certificate,
the expected payload SHA-256, and the expected $p$, prints exactly one
line --- \texttt{ACCEPT <classification> coherent in [<lo>,<hi>]x1e-50}
or \texttt{REJECT <reason>} --- and exits 0 or 1.

\paragraph{What is trusted.} Nothing in the certificate is taken on faith
except two \emph{fenced hints} per spectral record: the rounded center
matrix $B$ and the candidate eigenbasis $V$. These are hints in the
strict sense --- they steer the computation toward a particular (nearly
diagonalizing) frame, but every consequence drawn from them is
re-established by the verifier's own integer inequalities: $B$ must be
Hermitian and within the recorded Frobenius distance of the exactly
recomputed block; $V$ must pass the Gram-defect, square-root-witness, and
polar-distance obligations of Lemma~\ref{lem:polar}; the conjugated
matrix $V^\dagger B V$ is recomputed exactly and its inflated Gershgorin
discs (Lemma~\ref{lem:gershgorin}) are rebuilt from scratch. A wrong hint
can only cause rejection or a wider (hence more conservative) enclosure
--- never a wrong accepted sign. Every other recorded field (radii,
discs, components, entropies, traces, scaled integers, classification) is
\emph{derived data}: the verifier recomputes it and requires exact
equality, or, where the field is a bound, requires the recorded value to
satisfy the checked inequality in the conservative direction.

\paragraph{Algorithm.} (1) \emph{Startup self-test}: before touching the
certificate, the verifier proves $367879/10^6 < 1/e < 367880/10^6$ from
its own directed logarithm bounds (Lemma~\ref{lem:entropy}'s series with
integer floor/ceil at $10^{-70}$); this exercises the entire log
machinery on a known constant and anchors the $1/e$ entropy cap. (2)
\emph{Strict parse}: canonical types, the bit and array caps above,
pinned schema, $n$, constants, expected $p$, and expected payload digest.
(3) \emph{Payload decode}: the 1472 payload bytes are decoded to IEEE-754
binary64 values by explicit integer bit-field arithmetic (sign, exponent,
mantissa) --- not by the platform's float parser --- and immediately
converted to exact dyadics; the state's norm and Gram determinant are
checked positive exactly. (4) \emph{Per-block replay}, for
$f = 23,\dots,45$ in order: the reference block and its partial trace are
recomputed in exact rational arithmetic from the decoded payload via the
formula of Theorem~E (this is the expensive step); the recorded trace
enclosure is checked; then for each of the two spectral records the hint
obligations above are checked, the discs are rebuilt, the deterministic
disc-merge into components is replayed and required to match, and every
atanh/log/entropy interval is recomputed by cleared-denominator integer
comparisons and required to match the recorded rationals and scaled
integers. (5) \emph{Assembly}: multiplicity-weighted directed sums are
accumulated as scaled integers, the weighted global trace is checked to
contain 1, and the six final scaled integers and the classification must
match exactly. The final sign check is one integer comparison:
\texttt{coherent\_lower > 0} (POS) or \texttt{coherent\_upper < 0} (NEG).

\paragraph{Determinism and adversarial testing.} The verifier's output is
byte-identical across runs and platforms for identical inputs. The
artifact includes a mutation suite of seventeen adversarial certificates,
each derived from a genuine one by a single targeted corruption (a
flipped payload byte, a widened disc radius presented as narrower, a
tampered scaled integer, a de-canonicalized encoding, a reordered block,
a forged classification, \dots), each of which must be rejected with the
specific check named in its construction. Against the final emitter and
verifier, all seventeen mutations are rejected with exactly the named
reason and the unmutated certificate is accepted (eighteen tests, 59
minutes, single suite run). Measured verifier runtime is 27--33 minutes
per endpoint on an Apple M4 Pro (four verifications running
concurrently) and 72 minutes on a passively cooled MacBook Air;
certificate emission by the instrumented evaluator took 34--38 minutes
per endpoint on the M4 Pro and 77 minutes for a single endpoint on the
Air. The verifier source
(\texttt{verifier/verify.py}) has SHA-256
\texttt{\small 59662e7094bec9fafd6e084668f09b34}\allowbreak\texttt{\small 4824692fb88869bad64382865140ebb7}.

\section{Reproducibility}
\label{app:repro}

\paragraph{Artifact layout.} The artifact is archived as a versioned
public deposit at
DOI~\href{https://doi.org/10.5281/zenodo.21962925}{10.5281/zenodo.21962925}
(19{,}961{,}426-byte tarball, SHA-256
\texttt{\small e365a001d99aa310163c606ebce72171}\allowbreak\texttt{\small d680a99f5b7fc44cc53165d29afa17cc}).
It contains: \texttt{payloads/} --- the two raw 1472-byte state payloads,
SHA-256 as in Section~\ref{sec:setting}, together with
\texttt{payloads/notebook/} holding both states as
\texttt{torch.load}-ready tensors in the Ref.~\cite{agarwal2026}
repository's own format; \texttt{certificates/} --- the four endpoint
certificates of Appendix~\ref{app:format};
\texttt{verifier/verify.py} --- the standalone verifier of
Appendix~\ref{app:verifier}; \texttt{evaluator/} --- the exact evaluator
source (SHA-256
\texttt{\small ce75b1df1b327726e58a92b8926c9843}\allowbreak\texttt{\small cd9e668176340799a905ec5c3cd1ae34}),
which is \emph{not} needed for verification and is included only for
inspection and re-derivation; and \texttt{MANIFEST.txt} listing every
file with its SHA-256. The verifier SHA-256 is
\texttt{\small 59662e7094bec9fafd6e084668f09b34}\allowbreak\texttt{\small 4824692fb88869bad64382865140ebb7};
the four certificate SHA-256s are
\texttt{\small 795e55f3aa8bb4e1fd7e19f483b76b5b}\allowbreak\texttt{\small 1828cd83377736c009f253f410cb6964}
($R$, $p = 4057/62500$),
\texttt{\small 4866b9561d455c35efb08cd4cea651aa}\allowbreak\texttt{\small 3e9fce6314cdd776cfb3ed58829ce8b2}
($R$, $16229/250000$),
\texttt{\small 24a7fecc6cfdd5840696ac1d30220636}\allowbreak\texttt{\small 05d6e32544088b2a7a41733e4b783212}
($A$, $16239/250000$), and
\texttt{\small f15738747ca57b1e9f4d9b3fe5b2e35b}\allowbreak\texttt{\small 7e4ed6b04c3cc99145bed3db7f83413c}
($A$, $203/3125$).

\paragraph{The one verification command.} For each endpoint,
{\footnotesize
\begin{verbatim}
python3 verifier/verify.py \
    certificates/<endpoint>.json \
    <payload_sha256> <p_num>/<p_den>
\end{verbatim}
}
with the four (payload, $p$) pairs from the table in
Fact~\ref{fact:signs}. Four \texttt{ACCEPT} lines reproduce
Fact~\ref{fact:signs}, and with it --- given the three lemmas of
Appendix~\ref{app:lemmas} and the monotonicity of
Section~\ref{sec:monotonicity} --- every claim in this paper. No other
software, no network access, and no floating-point trust is involved.

\paragraph{Five-minute floating-point cross-check.} Independently of the
certificates, the public evaluator of Ref.~\cite{agarwal2026} can be
pointed at the shipped payloads: load either payload as 184 little-endian
binary64 values --- 92 complex amplitudes, the two rows of the rank-two
purification over the 46-dimensional $n = 45$ Dicke basis (also provided
in the repository's notebook-native format in the artifact) --- evaluate
the one-shot coherent information at the four endpoint values of $p$, and
compare against the certified enclosures of Fact~\ref{fact:signs}.
Floating-point agreement to several digits is expected and is \emph{not}
part of our proof --- it is a quick plausibility check that the two
codebases are looking at the same states and the same channel convention.

\section{Derivation of the block formula}
\label{app:derivation}

This appendix proves, from first principles, the block decomposition that
Section~\ref{sec:certified} states and the evaluator implements. Its only
citation is classical Schur--Weyl duality itself; everything
two-row-specific --- the Specht dimensions, the transfer matrix, the
coefficient extraction, the exact kernel with its normalization and index
conventions --- is derived here, so that the evaluator's arithmetic is
the literal implementation of Theorem~E below. Throughout,
$e_0 = |0\rangle$, $e_1 = |1\rangle$, logarithms are base two, and for a
positive (possibly subnormalized) matrix $A$ we write $\ent(A) =
-\Tr(A\log_2 A)$ with $0\log 0 := 0$.

\begin{elemma}[two-row Schur--Weyl decomposition]
\label{lem:schurweyl}
For $H_n = (\mathbb{C}^2)^{\otimes n}$,
$$H_n \cong \bigoplus_{f=\lceil n/2\rceil}^{n} U_{(f,n-f)} \otimes
S_{(f,n-f)},$$
and with $q = n - f$, $d = f - q = 2f - n$:
$$U_{(f,q)} \cong \det{}^{q} \otimes \mathrm{Sym}^{d}(\mathbb{C}^2),
\qquad \dim U_{(f,q)} = d + 1,$$
$$m_f := \dim S_{(f,q)} = \tbinom{n}{q} - \tbinom{n}{q-1} =
\tbinom{n+1}{f+1}\tfrac{d+1}{n+1}.$$
\end{elemma}

\begin{proof}
Classical Schur--Weyl duality decomposes $H_n$ over partitions of $n$
with at most two rows, i.e.\ $\lambda = (f,q)$ with $f + q = n$,
$f \ge q$. To identify $U_{(f,q)}$, let $\chi = (e_0\otimes e_1 -
e_1\otimes e_0)/\sqrt{2}$; for every $2\times 2$ matrix $T$,
$(T\otimes T)\chi = \det(T)\,\chi$. Hence the subspace $W_q =
\chi^{\otimes q} \otimes \mathrm{Sym}^d(\mathbb{C}^2) \subset H_n$
carries the representation $T \mapsto \det(T)^q\,\mathrm{Sym}^d(T)$. With
the collective lowering operator $L = \sum_{t=1}^n
|1\rangle\langle 0|_t$ (which kills $\chi$), the vectors
$$u_j = \sqrt{\tfrac{(d-j)!}{j!\,d!}}\; L^j\big(\chi^{\otimes q} \otimes
e_0^{\otimes d}\big) = \chi^{\otimes q} \otimes |D_j^d\rangle,$$
$0 \le j \le d$, form a basis of $W_q$ on which $\mathrm{diag}(t_0,t_1)$ acts with the
$d+1$ distinct weights $t_0^{\,q+d-j} t_1^{\,q+j}$, and raising/lowering
connects every adjacent pair; so $W_q$ is irreducible with highest weight
$(f,q)$, i.e.\ $W_q \cong U_{(f,q)}$, of dimension $d+1$.

For the multiplicity: the weight subspace of $H_n$ with exactly $q$
copies of $e_1$ has dimension $\binom{n}{q}$, and inside each
$U_{(n-t,t)}$ with $0 \le t \le q$ that weight occurs exactly once (the
vector with $q - t$ lowerings). Hence $\binom{n}{q} = \sum_{t=0}^{q}
\dim S_{(n-t,t)}$, and subtracting the same identity at $q-1$ gives
$\dim S_{(f,q)} = \binom{n}{q} - \binom{n}{q-1}$. (Combinatorially this
is the two-row ballot count: encoding a standard filling of $(f,q)$ by
its row word, the reflection that swaps symbols through the first
violating prefix is a bijection between bad words and words with $q-1$
second-row symbols.) The closed form follows from $\binom{n}{q} -
\binom{n}{q-1} = \binom{n}{q}\frac{n-2q+1}{n-q+1} =
\binom{n+1}{f+1}\frac{d+1}{n+1}$.
\end{proof}

\begin{elemma}[covariance and multiplicity-weighted entropy]
\label{lem:covariance}
Let $\Omega$ be an operator on $\mathbb{C}^2_{\mathrm{ref}} \otimes H_n$
that commutes with $I_{\mathrm{ref}} \otimes P_\pi$ for every permutation
$\pi \in S_n$. Then, in the decomposition of
Lemma~\ref{lem:schurweyl},
$$\Omega \cong \bigoplus_f B_f \otimes I_{m_f}$$
with $B_f$ acting on $\mathbb{C}^2_{\mathrm{ref}} \otimes U_{(f,n-f)}$;
if $\Omega$ is a density matrix then $\sum_f m_f\,\Tr B_f = 1$ and
$S(\Omega) = \sum_f m_f\,\ent(B_f)$; and $\Tr_{\mathrm{ref}}\,\Omega
\cong \bigoplus_f (\Tr_{\mathrm{ref}} B_f) \otimes I_{m_f}$, so
$S(\Tr_{\mathrm{ref}}\Omega) = \sum_f
m_f\,\ent(\Tr_{\mathrm{ref}} B_f)$.
\end{elemma}

\begin{proof}
Because the same channel acts on every factor, $\Dp^{\otimes n}$ commutes
with every $P_\pi$; Dicke vectors are permutation-fixed, so the joint
input and hence the joint output are $S_n$-invariant. Under Schur--Weyl,
$S_n$ acts as $\bigoplus_f I_{\mathbb{C}^2_{\mathrm{ref}} \otimes U_f}
\otimes \pi_f$ with the $\pi_f$ irreducible and pairwise inequivalent; an
invariant operator has no cross-sector blocks, and within a sector
Schur's lemma forces the Specht component to be scalar. The spectrum of
$B_f \otimes I_{m_f}$ is that of $B_f$ with every eigenvalue repeated
$m_f$ times, which yields the trace and entropy identities; the partial
trace over the reference acts inside $B_f$ only.
\end{proof}

\begin{elemma}[transfer matrix]
\label{lem:transfer}
With formal variables $x, y$ and $|x\rangle = e_0 + x e_1$, $\langle y| =
\langle 0| + y\langle 1|$:
\begin{multline*}
T(x,y) := \Dp(|x\rangle\langle y|)\\ = \begin{pmatrix}
1 - 2p + 2pxy & (1-4p)\,y \\ (1-4p)\,x & 2p + (1-2p)\,xy
\end{pmatrix},
\end{multline*}
$$\det T(x,y) = 2p(1-2p)(1+xy)^2.$$
\end{elemma}

\begin{proof}
Proposition~\ref{prop:affine} gives $\Dp(A) = (1-4p)A + 2p\,\Tr(A)\,I$;
applying it to $|x\rangle\langle y| = \big(\begin{smallmatrix} 1 & y \\
x & xy \end{smallmatrix}\big)$ with trace $1 + xy$ yields the matrix. For
the determinant put $z = xy$ and $\eta = 1 - 4p$: the constant and $z^2$
coefficients of $(1-2p+2pz)(2p+(1-2p)z) - \eta^2 z$ are both $2p(1-2p)$
and the $z$ coefficient is $(1-2p)^2 + 4p^2 - (1-4p)^2 = 4p(1-2p)$.
\end{proof}

\begin{elemma}[exact kernel by coefficient extraction]
\label{lem:kernel}
Fix a sector $f$; let $q = n - f$, $d = 2f - n$, and for final indices
$0 \le i, j \le d$ set $u = d - j$ and $v = d - i$ (the evaluator's
\texttt{raw\_row} and \texttt{raw\_column}). For any $2\times 2$ matrix
$T = \big(\begin{smallmatrix} A & B \\ C & D \end{smallmatrix}\big)$, the
normalized Dicke-basis matrix element of $\mathrm{Sym}^d(T)$ is
\begin{align*}
[\mathrm{Sym}^d(T)]_{ij} &= \sqrt{\frac{u!\,(d-u)!}{v!\,(d-v)!}}\;
P_{uv}(T),\\
P_{uv}(T) &= \sum_{s} \tbinom{v}{s}\!\tbinom{d-v}{u-s}
A^s\! B^{v-s} C^{u-s} D^{d-u-v+s},
\end{align*}
the sum over $\max(0, u+v-d) \le s \le \min(u,v)$. With $a_0 = 1-2p$,
$a_1 = 2p$, $d_0 = 2p$, $d_1 = 1-2p$, $\eta = 1-4p$, $\Delta =
2p(1-2p)$, define $R^{kl}_{ij} = 0$ unless $k - l = u - v$, and
otherwise, writing $b = v - s$, $c = u - s$, $e = d - u - v + s$ per
summand,
\begin{multline*}
R^{kl}_{ij} = \sum_s \binom{v}{s}\binom{d-v}{u-s}\,\eta^{\,b+c}\\
\times \sum_{\alpha=0}^{s}\binom{s}{\alpha} a_0^{s-\alpha} a_1^{\alpha}
\sum_{\beta=0}^{e}\binom{e}{\beta} d_0^{e-\beta} d_1^{\beta}
\binom{2q}{h}\,\Delta^{q},
\end{multline*}
where $h = k - (\alpha + \beta + c)$ and a summand is included only when
$0 \le h \le 2q$ and $l = \alpha + \beta + h + b$. Then the $f$-sector
block of $\Dp^{\otimes n}(|D^n_k\rangle\langle D^n_l|)$ has entry
$$K_f(i,j;k,l) = R^{kl}_{ij}\,
\sqrt{\frac{u!\,(d-u)!}{v!\,(d-v)!\,\binom{n}{k}\binom{n}{l}}}.$$
\end{elemma}

\begin{proof}
The coherent-vector identity $(e_0 + z e_1)^{\otimes d} = \sum_r
\sqrt{\binom{d}{r}}\, z^r |D_r^d\rangle$ gives
$$\sqrt{\tbinom{d}{j}}[\mathrm{Sym}^d(T)]_{ij} =
\sqrt{\tbinom{d}{i}}[z^j](A\! +\! Bz)^{d-i}(C\! +\! Dz)^{i}.$$
Choosing $s$ copies of $A$ from the first factor makes the powers of
$A, B, C, D$ equal $s$, $d-i-s$, $d-j-s$, $i+j-d+s$; substituting
$u = d-j$, $v = d-i$ yields $P_{uv}$ and, since
$\sqrt{\tbinom{d}{i}/\tbinom{d}{j}} = \sqrt{u!(d-u)!/(v!(d-v)!)}$, the
normalization --- including the evaluator's reversed-and-transposed raw
indexing. By Lemma~\ref{lem:schurweyl} the $f$-sector of
$T(x,y)^{\otimes n}$ is $\det(T(x,y))^q\,\mathrm{Sym}^d(T(x,y))$.
Expanding each factor by Lemma~\ref{lem:transfer} --- $A^s = \sum_\alpha
\binom{s}{\alpha} a_0^{s-\alpha} a_1^{\alpha} (xy)^{\alpha}$, $D^e =
\sum_\beta \binom{e}{\beta} d_0^{e-\beta} d_1^{\beta} (xy)^{\beta}$,
$B^b C^c = \eta^{b+c} y^b x^c$, and $\det(T)^q = \Delta^q \sum_h
\binom{2q}{h} x^h y^h$ --- produces monomials
$x^{\alpha+\beta+c+h}\,y^{\alpha+\beta+b+h}$, which contribute to
$x^k y^l$ exactly under the two stated conditions; comparing exponents
also forces the selection rule $k - l = c - b = u - v$. This is
term-for-term the evaluator's rational coefficient routine. Finally, from
$|x\rangle^{\otimes n}\langle y|^{\otimes n} = \sum_{k,l}
\sqrt{\binom{n}{k}\binom{n}{l}}\, x^k y^l\,
|D^n_k\rangle\langle D^n_l|$ and linearity of the channel, the sector
block of $\Dp^{\otimes n}(|D^n_k\rangle\langle D^n_l|)$ is the $x^k y^l$
coefficient divided by $\sqrt{\binom{n}{k}\binom{n}{l}}$, giving $K_f$.
(Coefficient extraction here is exact polynomial algebra: both variables
have degree at most $n = d + 2q$, so the roots-of-unity contraction
$[x^k y^l] F = (n+1)^{-2}\sum_{r,t} \zeta^{-rk}\zeta^{tl}
F(\zeta^r, \zeta^{-t})$ with $\zeta^{n+1} = 1$ recovers each coefficient
with no interpolation or approximation.)
\end{proof}

\begin{elemma}[reference side]
\label{lem:reference}
Let $|\Psi\rangle = N^{-1/2}\sum_{r=0}^1 |r\rangle_{\mathrm{ref}}
|v_r\rangle$ with $|v_r\rangle = \sum_k \psi[r][k]\,|D^n_k\rangle$ and
$N = \sum_{r,k} |\psi[r][k]|^2 > 0$. In sector $f$ the joint block on
$\mathbb{C}^2_{\mathrm{ref}} \otimes U_f$ has dimension $2(d+1)$ and
entries
$$[B_f]_{(r,i),(c,j)} = \frac{1}{N}\sum_{k,l=0}^{n} \psi[r][k]\,
\overline{\psi[c][l]}\; K_f(i,j;k,l),$$
and with the flattened index $r(d+1)+i$,
$[\Tr_{\mathrm{ref}} B_f]_{ij} = \sum_r [B_f]_{(r,i),(r,j)}$.
\end{elemma}

\begin{proof}
Expand $|\Psi\rangle\langle\Psi| = N^{-1}\sum_{r,c,k,l}
\psi[r][k]\overline{\psi[c][l]}\;|r\rangle\langle c| \otimes
|D^n_k\rangle\langle D^n_l|$; the identity channel preserves
$|r\rangle\langle c|$, Lemma~\ref{lem:kernel} supplies the
channel-sector entry, and linearity assembles the block. The partial
trace keeps $r = c$ and sums.
\end{proof}

\begin{etheorem}[exact block formula]
Under the hypotheses of Lemma~\ref{lem:reference}, $\Omega =
(\id_{\mathrm{ref}} \otimes
\Dp^{\otimes n})(|\Psi\rangle\langle\Psi|)$ is unitarily equivalent to
$\bigoplus_f B_f \otimes I_{m_f}$ over $f = \lceil n/2\rceil,\dots,n$,
with $d = 2f - n$, $m_f = \binom{n+1}{f+1}(d+1)/(n+1)$, and
\begin{multline*}
[B_f]_{(r,i),(c,j)} = \frac{1}{N}\sum_{k,l} \psi[r][k]\,
\overline{\psi[c][l]}\; R^{kl}_{ij}\\
\times\sqrt{\frac{(d-j)!\,j!}{(d-i)!\,i!\,\binom{n}{k}\binom{n}{l}}}.
\end{multline*}
Consequently $\sum_f m_f \Tr B_f = 1$ and
$$I_c = \sum_f m_f\big[\ent(\Tr_{\mathrm{ref}} B_f) - \ent(B_f)\big].$$
\end{etheorem}

\begin{proof}
Lemma~\ref{lem:kernel} gives every channel-sector matrix element,
Lemma~\ref{lem:reference} assembles the reference blocks, and
Lemma~\ref{lem:covariance} supplies the direct sum, the multiplicities,
the normalization, and the entropy identities.
\end{proof}

The evaluator is this theorem verbatim: its exact-state routine supplies
$\psi$ and $N$ from the payload dyadics; its rational-coefficient routine
is $R^{kl}_{ij}$; its entry kernel multiplies by the displayed square
root (enclosing the generally irrational root in a directed rational
interval); its reference block uses the flattened reference-first index;
its partial trace performs the Lemma~\ref{lem:reference} sum; and its
assembly loop weights both entropies by $m_f$. The derivation has
additionally been machine-checked term-for-term against the evaluator
source and validated by brute-force spectral comparison against direct
$2\cdot 2^n$-dimensional computation for $n \le 6$ (seventeen trials,
machine precision, including at $p = 16239/250000$).

\end{document}